\documentclass[prodmode,acmec]{ec-acmsmall}
\acmVolume{X}
\acmNumber{X}
\acmArticle{X}
\acmYear{2015}
\acmMonth{2}
\usepackage[numbers]{natbib}
\usepackage[numbers]{natbib}

\usepackage[ruled,vlined]{algorithm2e}
\renewcommand{\algorithmcfname}{ALGORITHM}
\SetAlFnt{\small}
\SetAlCapFnt{\small}
\SetAlCapNameFnt{\small}
\SetAlCapHSkip{0pt}
\IncMargin{-\parindent}

\usepackage{float}
\usepackage{epsfig}
\usepackage{amsmath}
\usepackage{amssymb}
\usepackage{amstext}
\usepackage{amsmath}
\usepackage{xspace}
\usepackage{latexsym}
\usepackage{verbatim}
\usepackage{multirow}
\usepackage{graphicx}
\usepackage{ifthen}
\usepackage{hyperref, url}
\usepackage{enumitem}
\usepackage{balance}

\usepackage{color-edits}
\addauthor{renato}{red}	
\addauthor{moshe}{blue}
\addauthor{balu}{green}

\newcommand{\abs}[1]{\lvert{#1}\rvert}
\newcommand{\fall}[1]{\underbar{r}}
\newcommand{\R}{{\mathbb{R}}}
\newcommand{\Z}{{\mathbb{Z}}}
\newcommand{\argmax}{\operatorname{argmax}}
\newcommand{\argmin}{\operatorname{argmin}}
\newcommand{\numSell}{s}

\newcommand{\pur}{X}
\newcommand{\puri}[1][i]{{\pur_{#1}}}

\newcommand{\supp}{n}
\newcommand{\suppi}[1][i]{{\supp_{#1}}}

\newcommand{\suppmax}{\supp_{\max}}

\newcommand{\Supp}{N}
\newcommand{\Suppi}[1][i]{{\Supp_{#1}}}

\newcommand{\util}{u}

\newcommand{\price}{p}
\newcommand{\prices}{{\mathbf \price}}

\newcommand{\pricej}[1][j]{{\price_{#1}}}
\newcommand{\pricesi}[1][i]{{\prices_{#1}}}

\newcommand{\br}{\textsc{Br}}
\newcommand{\disc}{\epsilon}

\newcommand{\mon}{\mu}
\newcommand{\moni}[1][i]{{\mon_{#1}}}

\newcommand{\ncomment}[1]{}

\newcommand{\val}{v}
\newcommand{\V}{{\mathbf{V}}}
\newcommand{\B}{{\mathbf{B}}}
\newcommand{\E}{{\mathbb{E}}}
\newcommand{\one}{{\mathbf{1}}}

\newtheorem{claim}{Claim}
\newtheorem{fact}[theorem]{Fact}

\newenvironment{proofof}[1]{{\textsc{Proof of #1.}}}{}

\makeatletter
\newenvironment{oneshot}[1]{\@begintheorem{#1}{\unskip}}{\@endtheorem}
\makeatother

\begin{document}
\markboth{M. Babaioff, R. Paes Leme and B. Sivan}{Price Competition, Fluctuations and Welfare Guarantees}
\title{Price Competition, Fluctuations and Welfare Guarantees}
\author{Moshe Babaioff
\affil{Microsoft Research. \tt{moshe@microsoft.com}.}
Renato Paes Leme
\affil{Google Research. \tt{renatoppl@google.com}.}
Balasubramanian Sivan
\affil{Microsoft Research. \tt{bsivan@microsoft.com}.}
}
\begin{abstract}
In various markets where sellers compete in price, price oscillations are observed rather than convergence to equilibrium. Such fluctuations have been empirically observed in the retail market for gasoline, in airline pricing and in the online sale of consumer goods. Motivated by this, we study a model of price competition in which equilibria rarely exist. We seek to analyze the welfare, despite the nonexistence of equilibria, and present welfare guarantees as a function of the market power of the sellers.

We first study best response dynamics in markets with sellers that provide a homogeneous good, and show that except for a modest number of initial rounds, the welfare is guaranteed to be high. We consider two variations: in the first the sellers have full information about the buyer's valuation. Here we show that if there are $n$ items available across all sellers and $n_{\max}$ is the maximum number of items controlled by any given seller, then the ratio of the optimal welfare to the achieved welfare will be at most $ \log\left(\frac{\supp}{\supp-\suppmax + 1}\right)+1$. As the market power of the largest seller diminishes, the welfare becomes closer to optimal. In the second variation we consider an extended model in which sellers have uncertainty about the buyer's valuation. Here we similarly show that the welfare improves as the market power of the larger seller decreases, yet with a worse ratio of $\frac{\supp}{\supp-\suppmax + 1}$. Our welfare bounds in both cases are essentially tight. The exponential gap in welfare between the two variations quantifies the value of accurately learning the buyer's valuation in such settings.

Finally, we show that extending our results to heterogeneous goods in general is not possible. Even for the simple class of $k$-additive valuations, there exists a setting where the welfare approximates the optimal welfare within any non-zero factor only for $O(1/s)$ fraction of the time, where $s$ is the number of sellers.
\end{abstract}
\acmformat{Babaioff, M., Paes Leme R., and Sivan B. 2014. Price Competition, Fluctations and Welfare Guarantees.}

\maketitle{}

\section{Introduction}
\label{sec:intro}
Price fluctuations have been observed in a variety of markets: from
the traditional retail market for gasoline (\citet{gasoline_1, gasoline_2}) all
the way to novel online marketplaces (\citet{EdelmanO2007}).
A recent WSJ article \cite{wsj_article} tracks the price of a microwave
across different online retailers (Amazon, Best Buy and Sears) and observes
sellers constantly adjusting prices leading to fluctuations.  The article
remarks that frequent price adjustments which used to be confined to domains
such as airline and hotel pricing are becoming increasingly common for all
sorts of consumer goods.
In this paper, we seek to analyze the efficiency
of such markets despite the lack of convergence to equilibrium. One interesting
feature is that the usual tools used in Price of Anarchy analysis (such as the
smoothness framework) are not available and instead one needs to directly
analyze the dynamics that leads to price fluctuations.

We consider a scenario with multiple sellers, each holding a set of goods, where
each seller's strategy is to set a price for each of his goods.
After prices are set, a buyer with a given valuation function over sets of goods (possibly
representing the aggregate demand of many buyers), chooses an optimal bundle of
items.
Sellers move sequentially and in each time period one of them responds to
the demand and the other sellers' prices by posting prices that are myopically
optimal. The supply for each seller refreshes every time period to the seller's original set.
%, and items are perishable.
In this paper we first study the special case of homogeneous goods, in
which all items are identical and then explore the case of heterogeneous goods.

As an illustration consider web-publishers. Each receives a fixed
number of impressions per day and sells those via posted prices\footnote{in
practice those are sold via an auction where the publisher can set the reserve price. Using the
reserve as a posted price, however, is a good first approximation to this
scenario.}. After prices are posted, a DSP (Demand Side Platform, typically a
network representing many advertisers), acquires impressions according to the
aggregate demand expressed by the advertisers. Each publisher (seller) can
observe the supply sold in each day and adjust his price. We are interested in
the dynamic that arises from each seller repeatedly updating his price to best
respond to other sellers. Below we discuss other markets
{with similar characteristics:}

\begin{itemize}
\item \emph{retail market for gasoline}: in each day a gas station can
sell a fixed amount of gas that it is capable of storing in its tank.
Periodically, stations update their price in response to the observed
consumption of the previous period as well as the prices posted by other gas
stations.
\item \emph{market for electronic components}: companies manufacturing
phones and other electronic devices typically buy component parts (such as
chips or flash memory) from other companies specialized in those. Typically
those suppliers update prices to respond to the demand of the buyer and to other
sellers.
\item \emph{airline tickets}: prices of airline tickets are the prototypical
example of price fluctuation arising from fierce competition. According to
\cite{economist_volatility}, there are roughly $1.86$ million price updates
per year for flights between New York City and San Francisco (across all
airlines, days and fare classes).
\end{itemize}
Clearly any real world example is much more complex than our simplified model.
Electronic components might not be completely identical, thus not being perfect substitutes.
Airlines sell tickets for the same flight during a long period of time, and
demand distribution shifts over time. Moreover, they experience seasonality and
many other issues.  Our model abstracts away all these issues in order to gain
tractability and isolates the price competition aspect of such settings.

In modeling price updates we assume that each seller maintains a price for each
of his units and at any given point in time, an arbitrary seller is allowed to
change his prices. In this paper we assume that sellers use \emph{myopic
best-response}, i.e., they update their prices to optimize their revenue in
the next time period based on their belief about the valuation function of the buyer
and on the current prices of other sellers. In the first part of the paper, we
consider the \emph{Full Information case}, in which each seller is certain about
the valuation of the buyer.
{In the second, we consider the \emph{Uncertain Demand case} in which
the valuation is still fixed over time, yet the sellers do not know it exactly.
Rather, each seller has a belief expressed as a distribution over valuation
functions of the buyer (and this belief does not rule out the actual
valuation).}
In the Uncertain Demand case, we still assume that each  seller
myopically responds to
the prices of the other sellers given his current belief about the demand,
{which is formed from his initial belief (prior) using all observations from
previous time-steps.}

{Myopic best response dynamics has been extensively studied for many repeated
game settings. It is attractive as a model of situations in which agents act
rationally in the short run, but lack deep understanding of the game and the
implications of their behavior in the long run. There is no coordination issue
and the dynamics is distributed in nature. In the Uncertain Demand case we still
assume that sellers are completely myopic and disregard the long term
implications of learning. Nevertheless, once observing new information about the
demand, they do update their belief. The myopic model allows us to focus on the
features arising purely from competition.}

\paragraph{Our results and techniques for homogeneous goods}
{Our paper relates to the model of price competition studied in
\citet{Babaioff14}.}
In that paper the authors note that if each seller holds only one item, {an
efficient} pure Nash equilibrium (NE) always
exists for any combinatorial valuation of the buyer (even when items are
not identical).\footnote{When the valuation is monotone there is always an
equilibrium in which all items are being sold.} They also show that if the
valuation of the buyer has
decreasing marginals and one seller holds all items
(monopolist case), then an equilibrium also exists but it can have welfare
$\Omega(\log n)$ factor away from the welfare of the optimal allocation, where
$n$ is the number of items.

If each seller holds multiple items but no seller is a monopolist, however, a
pure Nash equilibrium might not exist
(actually, we show that it fails to exist
unless there is an efficient Nash equilibrium, which occurs only in
very restricted settings). We illustrate in Figure~\ref{fig:price_cycle} a best
response sequence that cycles.
Pure Nash equilibrium might fail to exist even if the items are homogeneous, i.e., the valuation of a buyer
just depends on the number of items he acquires. This fact is in line with
observations that for various markets that
{can be approximately captured by}
this model, prices fluctuate
rather than converge to an equilibrium.

\begin{figure}
\centering
\includegraphics[scale=.45]{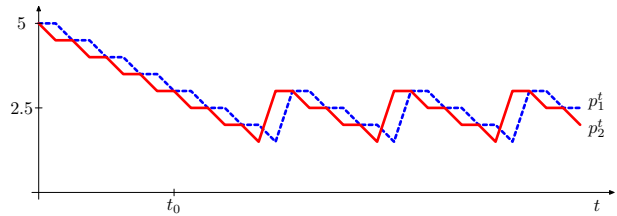}
\caption{The time sequence of prices posted by sellers $1$ and $2$, each with 2 items, when the buyer have a valuation with marginals $5,5,3,1$.
% Example \ref{example:no-Nash-eq}.
Notice that a best-response sequence cycles forever rather than converge to an equilibrium.}
\label{fig:price_cycle}
\end{figure}

Despite the lack of equilibrium, we still want to provide \emph{welfare
guarantees} for this setting. The guarantees we provide are as follows: we
say that a game has \emph{Eventual Welfare Guarantee} (EWG) $\alpha$ if
{for any initial prices and}
%for
every sequence of best-responses, after finitely many \emph{rounds} of price
updates (where a round is a sequence of price updates in which every seller
updates his price at least once (\citet{NisanSVZ11})),
the welfare is at least an $1/\alpha$ fraction of the optimal welfare
\emph{in every time-step}.
In particular, we show that if there is a total supply of $n$ items and
$n_{\max}$ is the maximum number of items controlled by any given seller, then
the EWG is smaller than $1+\log (\supp/ (\supp-\suppmax +1))$%, \mbedit{and this is essentially tight}
.\footnote{All logarithms in this paper are natural
logarithms (base $e$).}
{This high welfare is reached quite fast, after number of rounds that is
linear in the number of sellers and the inverse of a discretization parameter.
An immediate corollary is that} if no seller controls more than a $c$ fraction
of the goods, the EWG is {smaller than
$1-\log (1-c)$. These results present a bound on the welfare loss as a function
of the market power of the seller with the largest supply.
As that seller controls a smaller share, efficiency increases (as expected), not
only at equilibrium but also when prices fluctuate (after enough price
updates).}
For the Uncertain Demand case, we show the exponentially worse EWG %\mbedit{tight}
bound of $\supp/ (\supp-\suppmax +1)$. If every seller controls at most a $c$ fraction of the
market, this implies an EWG of $1-c$. We also show that both bounds on the
Eventual Welfare Guarantee are tight up to constant factors.

\paragraph{Our results for heterogeneous goods} The dynamics changes quite dramatically in the setting of heterogeneous goods.
% MB 10.21.14: many changes  in the paragraph below
When the goods are heterogeneous, there is a whole spectrum of buyer valuations, some very complicated. One would naturally be interested in understanding simple valuation first and then move to more complicated ones.
Two extreme valuations that are simple yet incomparable are unit-demand and additive valuations.
%starting from unit-demand, all the way till additive valuations.
% The two extreme cases of unit-demand and additive valuations
These two extremes are well behaved: they both have an efficient pure Nash equilibrium. Given this, we move slightly away from these extremes to explore the simplest valuations that generalize both additive and unit-demand valuations, which are
$k$-additive valuations: given a bundle, the buyer's valuation is additive over the $k$ most valuable items in the bundle.

First, we show that for $k$-additive valuations an efficient pure Nash equilibrium need not always exist. Interestingly, in a heterogeneous goods setting, we observe that increasing competition by splitting the goods held by a monopolist among different sellers could hurt welfare! Whenever a pure Nash  equilibrium exists, the welfare obtained is very close to optimal welfare, and we obtain a tight characterization of welfare in any pure Nash equilibrium. But when it doesn't exist, the welfare obtained is very poor in a strong sense: we show that the Eventual Welfare Guarantee is unbounded. Further, there exist best response sequences with $s$ sellers where a non-zero approximation to optimal welfare is obtained only in $O(\frac{1}{s})$ fraction of the time. This implies that the positive results for the homogeneous case cannot be extended to heterogeneous items, even when restricting to the relatively simple $k$-additive valuations.

\paragraph{Relation to other equilibrium concepts}

Before we proceed, we would like to comment on our choice of  Eventual Welfare
Guarantee rather than more traditional alternatives. A common choice in settings
where pure Nash equilibria don't exist is to study mixed Nash equilibria, in
which each agent plays according to a distribution. A mixed Nash would be
characterized by a setting in which each agent has a distribution over actions
(in our case, each seller has a distribution over prices) and in each round he
plays an independent sample from this distribution. Although a
game-theoretically sound concept, it doesn't seem to reflect the empirical
behavior of sellers in various markets, since the prices used in a certain time
period are very related to prices used in the previous one. {This is unlike the
case of sellers sampling independently from their mixed strategies, which
results in frequent large changes of prices, unrelated to recent prices.} The
same argument can be made about other static equilibrium notions such as
correlated and coarse correlated equilibrium.

A popular dynamic alternative is to consider outcomes of no-regret dynamics --
see \citet{nadav10} and \citet{immorlicaMP10} for a discussion of no-regret
outcomes in price competition and  \citet{Roughgarden09,SyrgkanisT13} for
examples of this analysis in various other settings. While those model
strategy updates that are dependent on the outcome of previous steps, the
analysis generates welfare guarantees \emph{on average}, i.e., if agents play
long enough, the average welfare is at least a certain factor of the optimal
welfare.  Our analysis provides guarantees that hold for {\em each} time period
after a certain point. The concept of \emph{sink equilibria} proposed by
\citet{GoemansMV05} is closely related yet weaker, it only provides {\em
average} guarantees over the Markov chain defined by the best response graph.
%Our analysis provides guarantees that hold for each time period after a certain
%point. The concept of \emph{sink equilibria} proposed by \citet{GoemansMV05}
%also provides average guarantees over the Markov chain defined by the best
%response graph.

A main difference between our approach and the
no-regret learning/sink-equilibria approach is that in the latter, the
welfare bounds are studied by analyzing a static \emph{limit object}. For the
no-regret learning, the limit object is typically a
coarse correlated equilibrium and for sink equilibria, it is the stationary
state of a certain Markov chain. One drawback of such approaches is that
although they provide a good description of the limiting behavior of the
process, they don't offer intuitions on the transient behavior of the dynamic.

\paragraph{Related work}
The work in out-of-equilibrium versions of price competition started by the work
of Edgeworth \cite{Edgeworth_monopoly} who criticized the prediction in the
oligopoly models of Cournot and Bertrand that prices converge to an equilibrium.
Edgeworth pointed out  that various small changes in model caused the
equilibrium no longer to exist. This idea was later developed by Shubik
\cite{shubik-book}, Shapley \cite{shapley57} and others. We refer to
\cite{vives} for a survey on the modern theory of oligopolies. Various early
concepts on oligopolies and price fluctuations received a formal treatment in a
sequence of papers by Maskin and Tirole
\cite{maskin_tirole_1,maskin_tirole_2,maskin_tirole_3}.
%{Our paper studies similar models but does so} through the lens of Price of
%Anarchy analysis.
Our paper studies similar models but seeks to quantify the welfare loss from the
resulting dynamic on a per-step basis instead of characterizing equilibria. 

Another stream of related work is the study of the outcome of best response
dynamics in algorithmic game theory: \citet{MirrokniV04} study the relation
between Nash equilibria and outcomes of best response dynamics. They show that
even when all pure and mixed Nash equilibria have high welfare, there can
exist sequences of best response dynamics which cycle through states of very low
welfare. \citet{Roughgarden09} shows that for the special case of potential
games, the smoothness framework can be used to provide efficiency guarantees for
best-response dynamics. Recently, \citet{FanelliFM12} study the rate in which
best-response dynamics converge for potential games.

Also related is the work of \citet{NisanSVZ11} who characterize a class of games
for which it is rational for agents to best respond. The authors also look at
the same question from the mechanism design perspective and seek to design games
that converge to desirable outcomes. More recently, \citet{EngelbergFSW2013}
give conditions for games converge under any sequence of best responses.

Another popular alternative, that relaxes best-response dynamics and models
bounded rationality and bounded computational abilities, is logit
dynamics~\cite{McFadden74}. Here, at each time-step, one among the $n$ agents in
the game is chosen uniformly at random, and the chosen agent $i$ plays a
strategy $s$ with probability proportional to $e^{\beta u_i(s,s_{-i})}$, where
$\beta$ is a parameter that models the level of rationality, with higher $\beta$
being more rational, and $u_i(\cdot)$ is the utility of agent $i$. This dynamics
defines a Markov chain over the strategy profiles, and this Markov chain has a
unique stationary distribution, which is called as the logit
equilibrium~\cite{AFPP13}. Given this, the primary questions that are studied
here are the rate at which the dynamics converges to this logit
equilibrium~\cite{AFPPP11}, the existence of metastable distributions in cases
where the dynamics takes a long time to converge~\cite{AFPP12}, what happens
if players can sumultaneously update their strategies~\cite{AFPPP12, AFPPP13}.
Although the literature in this area is taking steps to explore the transient
dynamics in cases where the dynamics takes a long time to converge~\cite{FV12,
AFPP12}, the main difference from our approach is that our work explicitly
studies the transient dynamics and gives guarantee on a per-round basis, rather
than the stationary or metastable distribution that is accomplished in time. 

While in this paper we study the welfare guarantees in presence of price
competition, one could also study seller revenue in presence of price
competition. See~\cite{GKT99, GK99, KG99} for a one-shot version Nash equilibrium
anaylsis of a price competition game between sellers, using various price update
strategies.

\section{Preliminaries}
\label{sec:prelim}
\subsection{Homogeneous goods}
\paragraph{Pricing game with full information} A full information pricing game is
defined by $\numSell\geq 2$ sellers, with each seller $i\in [\numSell]$ holding $\suppi$ units of a homogeneous good {for which he has no value},
and a buyer with a publicly known valuation $\val:\{0,1,\hdots,
\supp\} \rightarrow \R_+$, where $\supp := \sum_{i=1}^{\numSell} \suppi$ denotes
the total supply of the good.
{We define $\suppmax$ to be the maximal supply of any single seller,
$\suppmax=\max_{i=1}^s \suppi$.}
We assume that the valuation function is
(i) monotone, i.e.,  $v(k) \leq
v(k+1), \forall k$; (ii) is normalized at zero, i.e., $v(0) = 0$ and (iii)
obeys the decreasing marginal returns property, i.e., $m_1 \geq m_2 \geq \hdots
m_n$ where we denote by $m_k = v(k) - v(k-1)$ the $k$-th marginal value.
We also use $v(\ell|k)$ to denote the marginal value of $\ell$ items
with respect to a set of $k$ items, i.e.,  $v(\ell|k)= v(k + \ell)-v(k)$.

Let $\Suppi =\{1+\sum_{r=1}^{i-1}\suppi[r],\dots,\suppi+\sum_{r=1}^{i-1}\suppi[r] \}$ denote
the set of $\suppi$ units that seller $i$ holds. Let $\pricej$ denote the price
of item $j$, and let $\pricesi$ denote the vector of prices posted by seller
$i$. The utility of seller $i$ is his total revenue, given by
$\util_i^{\pur^v}(\prices) = \sum_{j \in \puri^v(\prices)}\pricej$ where
$\puri^v(\prices)$ is the bundle of goods purchased by the buyer with valuation $v$ from seller
$i$, when facing price vector $\prices$.

We assume that the buyer responds to prices by picking a bundle maximizing
his total value minus his payment. {Since the items are homogeneous and
the marginals decreasing, the optimal choice of the buyer is to follow a
\emph{Greedy Algorithm}\footnote{The fact that a greedy
algorithm computes the utility maximizing bundle for the buyer holds more
generally when the items are heterogeneous and the buyer has a gross substitutes
valuation. See~\citet{gs_survey}.}: start with an empty set of items and
recursively choose one of the cheapest items available. If the price of this item is smaller than the marginal value, include the item in the selected set and repeat.
% MB 10.21.14: added:
If the price is larger, then stop.
If the price of this item equal to the marginal value, either include the item in the selected set and repeat, or stop.
Any specific Greedy Algorithm needs to determine how to break ties between items of the same price,
as well as  deciding when to stop picking items for which the price is equal to the marginal value.
In fact, all optimal sets can be described as
outcomes of a greedy algorithm for some tie breaking rule and some rule of picking items with zero contribution to the utility of the buyer.
For clarity of presentation we fix the tie breaking rule used by the buyer as follows\footnote{See Appendix~\ref{app:randomTie} for a short discussion of random tie breaking.}:
between two items of the same price he breaks ties lexicographically (consistent with the lexicographic order over the sellers).
Also, he will always pick items with zero marginal utility, i.e., items for which the
price is equal to the marginal value.  We define
$\pur^v(\prices)$ as the set purchased by the buyer with value $v$
when facing a price vector $\prices$ according to the specific greedy algorithm discussed above.
When this is clear from the context, we omit $v$ and
refer to $\pur^v(\prices)$ as $\pur(\prices)$. }

We define the welfare associated with a strategy profile as the sum of utilities
of all agents involved (the buyer and all the sellers): $W(\prices) =
v(\vert X^v(\prices) \vert)$.

\paragraph{Discretization} We choose in this paper to model all valuations and
prices as integral multiples of a (smallest) monetary unit $\disc$ (e.g.,
$1$ cent in the case of US dollars).
{To simplify, we also assume that $1/\disc$ is an integer.}
We are primarily interested in studying price update
dynamics in this paper, where sellers frequently undercut
each other. A non-trivial change in price is necessary to meaningfully describe
the rate of progress of such update dynamics. Moreover, the constraint of
pricing in multiples of a fixed quantity is present in
essentially every (electronic or not) commercial platforms.

\paragraph{Notation}
We refer to the discretized domain as $\epsilon \Z_+ := \{0, \epsilon, 2
\epsilon, 3 \epsilon, \hdots \}$.
We assume from this point on that for all $k \in \{0,1,\hdots, n\}$,
$v(k) \in \epsilon \Z_+$ and that the strategy of a seller consists of choosing
a vector  $\prices_i \in \epsilon \Z_+^{n_i}$.

We will refer to the set of valuation functions as $\V$, i.e., $\V$ is the set
of valuation $v: \{0,1,\hdots, n\} \rightarrow \epsilon \Z_+$ that satisfy (i)
monotonicity, (ii) normalization at zero and (iii) decreasing marginal values.

\paragraph{Eventual Welfare Guarantee} The traditional solution concept for full
information games is that of the pure Nash equilibrium: a price vector is an
equilibrium if no seller can change the prices of the goods he controls and
improve his utility. Formally:\\
${\bf Nash}^{\pur^v} = \{\prices\in \epsilon
\Z_+^n:\util_i^{\pur^v}(\pricesi,\pricesi[-i])
\geq \util_i^{\pur^v}(\tilde{\prices}_{i},\pricesi[-i]),
\forall i\in
[\numSell],
\tilde{\prices}_{i}\in \epsilon \Z_+^{n_i}\}.$\\
As we show in Section~\ref{sec:fullInfo}, a pure Nash equilibrium does
not always exist in the pricing game defined above.
Given the absence of a pure strategy Nash equilibrium, and the motivation to
model price fluctuations
observed in online markets, we propose a notion on how to measure welfare
which we call \emph{Eventual Welfare Guarantee}.

First, we define a \emph{best response sequence}: We say that a sequence of
vectors of prices $\prices^0, \prices^1, \prices^2, \hdots, \prices^t, \hdots$ is a
{\em best-response sequence} if for each time-step $t$ there is a seller $i(t)$ such
that $\prices^t_{i(t)} \in  \br_{i(t)}^{X^v}(\prices^{t-1}_{-i(t)})$ and
$\prices^t_{i'} = \prices^{t-1}_{i'}$ for all $i' \neq i(t)$, where
$\br_i^{X^v}(\pricesi[-i])$ is the set of best-responses of seller $i$:
$$\br_i^{X^v}(\pricesi[-i]):= \argmax_{\pricesi \in  \epsilon \Z_+^{\suppi}}
\util^{\pur^v}_i(\pricesi,\pricesi[-i])$$
When the seller has more than one best-response, we will assume ties are broken
favoring uniform price vectors over discriminatory prices,
favoring the sale of more items and any remaining tie is broken according to an
arbitrary tie breaking rule. At the end of each time-step $t\geq 0$, the buyer
buys his favorite bundle $\pur^v(\prices^t)$, and
the sellers' supplies are replenished so that seller $i$ has $\suppi$ units.

We say that a best-response sequence is {\em fair}, if every seller is allowed to
best-respond infinitely often, i.e, for all $i' \in [s]$, $\{t; i(t) = i'\}$
is an infinite set. Given a fair best response sequence, we divide the time in rounds,
where a {\em round} is a minimal interval in which each seller has a chance to play at
least once. Formally, we define the beginning time $r_\ell$ of the
$\ell$-th round recursively:
$r_0 = 0$, and $r_\ell = \min \{r; \forall i' \in [s] \exists t \in [r_{\ell-1}, r) \text{ s.t. } i(t) = i' \}$.

We say  that a game has {\em Eventual Welfare Guarantee $\alpha$} if there is
a finite $\tilde{\ell}$ such that from round $\tilde{\ell}$ onwards, for every fair best-response sequence, the
welfare is at least a $\frac{1}{\alpha}$ fraction of the optimal welfare $W^* = v(n)$. Formally: for every $t \geq r_{\tilde{\ell}}$,
$W(\prices^t) \geq \frac{1}{\alpha} W^*$.

\paragraph{Pricing game with uncertain demand}
In the basic model we assumed that the valuation of the buyer is fixed and
known by all the sellers. In the uncertain demand model we still assume that the
buyer picks an optimal set with respect to some fixed true underlying valuation
$v^*$ but this valuation is not known to the sellers. Instead, each seller has a
(possibly different) belief expressed as a probability distribution over
possible valuations of the buyer. We assume the beliefs are consistent in the
sense that they assign non-zero probability to the true valuation $v^*$.

An important special case of the model is the case where there is a common prior
and the valuation of the buyer is drawn from this common prior at time-step zero
and used for all subsequent time-steps. In this special case, the prior beliefs
are consistent by definition.

We model the belief of seller $i$ at time $t$ as a probability distribution
$\B_i^t$ over $\V$, the set of all valuation functions. Given $v \in \V$, we
denote by $\B_i^t(v)$ the probability assigned
by seller $i$ to valuation $v$ at time $t$. We call $\B_i^0$, the prior of seller
$i$ and say that the prior is consistent if $\B_i^0(v^*) > 0$.
We observe that since $\V$ is a countable discrete set, a seller believing that
some valuation is feasible is equivalent to assigning positive probability to
it, unlike in continuous settings.

Given a sequence of prices $\prices^0, \prices^1, \prices^2, \hdots$,
in each
time-step $t$, the buyer faces prices $\prices^t$ and he
picks his optimal set $X^{v^*}(\prices^t)$ according to
his true valuation function $v^*$. Upon observing the set $S^t$ picked by the
buyer in time-step $t$, each seller $i$ updates his belief
{in a Bayesian way to take into account the purchasing information at time $t$}

$$\B_i^t(v) = \frac{ \one \{ S^t = X^v(\prices^t)\} \cdot \B_i^{t-1}(v)
}{\sum_{v' \in \V} \one \{ S^t = X^{v'}(\prices^t)\} \cdot \B_i^{t-1}(v')}$$
where $\one\{\cdot\}$ is the indicator function.
% MB 10.21.14:
% OLD: The update is well defined only in the case all priors are consistent.
% NEW:
The assumption that for every seller has a consistent prior guarantees that the update is well defined.

We say that the sequence of prices is a {\em best response sequence} if for each time-step $t$, there is an
seller $i(t)$ such that $\prices_{i(t)}^t \in \br_{i(t)}^{X,
\B_{i(t)}^{t-1}}(\prices^{t-1}_{-i(t)})$ and $\prices_{i'}^{t} =
\prices_{i'}^{t-1}$ for $i' \neq i(t)$, where:
$$\br_{i}^{X,\B_i^{t-1}}(\prices_{-i}) := \text{argmax}_{\prices_i} \E_{v \sim
\B_i^{t-1}} \left[ u_i^{X^v}(\prices_i, \prices_{-i}) \right]$$

In other words, the {best responding seller picks a price vector that optimizes
his expected
utility according to his belief and given the prices of all other sellers.
We stress that sellers are behaving myopically and are not explicitly trying to
learn the buyer valuation (say by exploring various prices that are not
currently optimal).}

Given the definition of a best response sequence in the uncertain demand
setting, we define a fair best response sequence and the concept of Eventual
Welfare Guarantee in the same way as for a full information game.

\subsection{Heterogeneous goods}
\paragraph{Pricing game with full information} We define a full information pricing game with heterogeneous goods analogously to that for homogeneous goods. Our terminology doesn't distinguish between different goods and multiple copies of the same good. There are $s$ sellers with each holding an arbitrary subset of the universe of goods. We assume the buyer's valuation function to be $k$-additive: his value for a bundle of goods is the sum of the values of his $k$ most valuable goods in that bundle. With a such a valuation function, the optimal choice of a buyer while responding to a vector of prices is to still follow a {\em Greedy Algorithm}: start with an empty set, and recursively pick the item that gives the largest non-negative utility (item's value minus its price), and continue till the set is of size $k$ or all items offering non-negative utility are exhausted. As before, while there are several ways of breaking ties in a Greedy Algorithm, for clarity we stick to the following tie-breaking rules for the buyer. If at any step in the greedy algorithm multiple items offer the same non-negative utility, the buyer prefers the item from the lexicographically most preferred seller. Among multiple items from the same seller that offer the same utility, the buyer prefers the most expensive item. The buyer buys items of zero utility (as long as that doesn't make his purchased set larger than size $k$). Any remaining tie is broken arbitrarily.

\section{Full Information Game for homogeneous goods}
\label{sec:fullInfo}
We start the section by showing  that pure Nash equilibrium might fail to exist
in the full information game, and then present our upper bound on the Eventual
Welfare Guarantee of this game. We conclude by showing that the bound is tight.

\subsection{Nonexistence of Pure Nash Equilibria}
Before proceeding to discuss Eventual Welfare Guarantee of the full information
pricing game, we show that pure Nash equilibria (NE) need not always exist for
such games. When every seller has only one unit ($s=n$), there is always
a pure Nash equilibrium in which every seller prices at the $n$-th marginal $m_n$ and sells his
unit. The next example shows this is no longer the case when sellers have at
least two units each, and this is true for any tie breaking rule. We prove this in Appendix~\ref{app:fullInfo}.

\begin{example}
\label{example:no-Nash-eq}
Consider a setting with two sellers, each holding 2 units of the same good.
The buyer has diminishing marginal valuations, of $5,5,3,1$, i.e.,
$\val(1) = 5$, $\val(2) = 10$, $\val(3) = 13$, $\val(4)=14$. This setting does not have any pure Nash equilibrium.
\end{example}

%Clearly, the number of units sold at any pure NE cannot be 1 because every
%seller can always make a non-zero revenue. It cannot be two either. There are
%two ways in which two goods can be sold. Both goods sold from the same seller,
%which is not possible at any equilibrium. The second way is to have just one
%good sold from each seller. This can happen only if the unsold item from each
%seller is priced above $3$. Either seller can price the unsold item at $3$ (and
%the sold item also at $3$) to get a revenue of $6$ which is larger than $5$
%which is the maximum they would have gotten previously. Next, the number sold at
%pure NE cannot be three too. Without loss of generality, we assume that this
%was reached with seller 1 selling two items and seller 2 selling one item. Let
%$p_1$ and $p_2$ denote the prices used by seller 1. Given this, seller 2 can
%obtain a revenue of at least $p_1 + p_2 -2\epsilon$ by pricing at $p_1-\epsilon$
%and $p_2-\epsilon$. If $p_1 + p_2 > 3+2\epsilon$, then seller 2 now gets a
%revenue strictly larger than $3$ which is the maximum he could have gotten by
%selling one item when three items were totally sold. On the other hand, if
%$p_1+p_2 \leq 3+2\epsilon$, seller 1 could have priced his items at $3$ and $1$
%and gotten a total revenue of $4 > 3+2\epsilon = p_1+p_2$. Finally, the number
%of goods sold at any pure NE cannot be 4 because that would imply that all the
%four units were priced at $1$, and one of the sellers can instead price both his
%items at $3$, sell one item and get a larger revenue of $3$ instead of $2$.
%\end{proof}
{In Figure \ref{fig:price_cycle} we depict a price cycle resulting from a
best-response sequence between the two sellers in the previous example. In
time-step $t$, seller $1$ posts a price of $p_1^t$ for both of his items and in the next time-step $t+1$, seller
two posts a price of $p_2^{t+1}$ for his items.
We assume that seller $2$ best-responds in odd time-steps and seller $1$ in even ones.
We set $\disc=0.5$ and break ties toward seller 1.
Note that within $8$ time-steps (marked as $t_0$ in the figure) $3$ units are sold,
and that the sold quantity will stay at $3$ forever in this example (although prices will cycle forever).
The welfare when $3$ items are sold is $13/14$ fraction of the optimal welfare (93\%).
For this example, our main theorem presents a guarantee of a slightly smaller
fraction: $(1+\log(4/3))^{-1}$ = 77\%.}

The fact that there is no pure Nash equilibrium for the above example is no
coincidence.
We next show that the full information game will rarely have a pure
equilibrium.
Unless there is an efficient equilibrium in which all units are sold at the
marginal price of the last unit $m_n$,
there is no pure Nash equilibria at all.
This is in line with Edgeworth's original observation
\citeyear{Edgeworth_monopoly} for a different oligopoly model.

\begin{proposition}
Assume that the buyer has a decreasing marginal valuation with
$m_n=\val(n)-\val(n-1) > 0$ and there are at least $2$ sellers. Fix an
arbitrary deterministic tie breaking rule for the buyer and assume that $\disc>0$ is small
enough. Either the only pure Nash equilibrium is for all sellers to
price at $m_n$ and sell all units (an efficient equilibrium), or there is no
pure Nash equilibrium.
\end{proposition}
\begin{proof}
First observe that in any efficient Nash equilibrium the price of every
seller must be $m_n$ as otherwise either the buyer will drop some units (if the
price is larger), or a seller can increase his revenue (if the price is lower).

Assume that $\disc<m_n/(4n)$.
Consider any pure Nash equilibrium which is not efficient, that is, there is at
least one seller $i$ that is not selling his entire supply. Assume that the
buyer is buying $q$ units in total.

We first observe that all items must be sold for about the same price (prices of
sold units can differ by at most $\disc$).
Assume that in a pure equilibrium the most expensive unit is sold at price $h$
(note that $h\geq m_n$).
Then any sold unit is priced either at $h$ or at $h-\disc$.
This is so as if some unit is priced at $p<h-\disc$ then the seller selling this
unit can increase his revenue by increasing the price of this unit to $h-\disc$
and still sell (as the buyer must be running a greedy algorithm as presented
above).

Next, observe that at any pure Nash equilibrium every seller is selling at least
one item and making positive revenue, as if a seller prices all his units at
$\disc<m_n$ the buyer will always buy from him. We conclude that other than
seller $i$  that is not selling all his units, there is at least one other
seller that sells at least one unit for price of at least $h-\disc$.
Consider a deviation by seller $i$ changing all the prices of his units to $h-2\disc$.
Since the buyer must be using a greedy algorithm, this seller must be selling
the minimum between $n_i$ and $q$, which is strictly larger than the quantity he
was selling before. His utility increased by at least $h-2\disc-2(n_i-1)\disc$
as he gains at least $h-2\disc$ for extra units sold, and loses at most
$2(n_i-1)\disc$ for units sold before and for which he has reduces his price.
As $\disc<m_n/(4n)$ it holds that $h-2\disc-2(n_i-1)\disc\geq m_n-2\disc n_i>
m_n-2n(m_n/(4n))>0$, thus this is a beneficial deviation. This contradicts the
fact that the this was a pure Nash equilibrium.
\end{proof}

\subsection{Bounding the Eventual Welfare Guarantee}

In this section we prove the main theorem in the paper, which is a bound on the
Eventual Welfare Guarantee of the full information game in the homogeneous goods setting. % our pricing mechanism.

\begin{theorem}\label{thm:main-epoa}
Assume that the buyer has a decreasing\\
 marginal valuation.
The Eventual Welfare Guarantee of the full information game is at most
$1+\log\left(\frac{\supp}{\supp-\suppmax + 1}\right)$. Moreover, such
welfare is achieved after at most $s \cdot v(1) / \epsilon$ rounds.
\end{theorem}

We note that this theorem generalizes Theorem 23 in \cite{Babaioff14}, which
gives an $O(\log n)$ bound for the Price of Anarchy for the monopolist case.

In what follows, we will prove a sequence of lemmas leading to our main result.
We remind the reader that the buyer upon being faced with a price vector
$\prices$, chooses a set $X(\prices)$ using the greedy algorithm with
tie breaking rules as discussed in Section \ref{sec:prelim}.

\begin{lemma}\label{lem:mostExpensive}
Let $\price^*(\val,\prices)$ denote the price of the most expensive item that a
buyer with value $\val$ purchased at prices $\prices$. Then, at any other price
vector $\tilde{\prices}$ that has at least $|\pur(\prices)|$ prices no larger
than $\price^*(\val,\prices)$, the number of items sold is at least
$|\pur(\prices)|$.
\end{lemma}
\begin{proof}
At prices $\prices$, the buyer obtained a non-negative utility from purchasing
his $|\pur(\prices)|$-th unit at $\price^*(\val,\prices)$. Therefore, because of
decreasing marginal utility, the utility obtained from buying $k$th unit for any
$k \leq |\pur(\prices)|$ at a price at most $\price^*(\val,\prices)$ is also
non-negative. Thus, the greedy procedure will not terminate before purchasing at
least $|\pur(\prices)|$ units.
\end{proof}

{The fact that the buyer purchases items following the greedy algorithm
gives the seller no incentive to price different units differently.}

\begin{lemma}\label{lem:unifPricing}
For every price vector $\pricesi[-i]$ of other sellers, every best-response
price vector of seller $i$ prices all his sold units at the same price.
\end{lemma}
\begin{proof}
Consider any candidate price vector $\pricesi$ for seller $i$ where $i$ makes a
non-zero revenue. Let $h$ be the
price of the most expensive unit that seller $i$ sells at
$(\pricesi,\pricesi[-i])$.
Consider raising the price of every unit of seller $i$ to $h$. By
Lemma~\ref{lem:mostExpensive}, the total number of sold units will not fall
down. The number of units $|\puri(\prices)|$ sold from $i$ also will not fall
down: for it to fall down, the buyer should have compensated by buying some
previously unsold unit from some seller other than $i$. But the greedy procedure
implies that every unsold unit from sellers other than $i$ at the original price
vector $\prices$, is either more expensive than $h$ or is priced at $h$ but held
by a lesser preferred seller. So $|\puri(\prices)|$ will not go down either.
Thus pricing all sold units at $h$ strictly increases seller $i$'s revenue (if
at least two of $i$'s units were sold at $(\pricesi,\pricesi[-i])$), and
this proves the lemma.
\end{proof}

Recall our assumption that sellers prefer uniform price vectors over non-uniform
price vectors and prefer selling larger number of items.
Lemma~\ref{lem:unifPricing} states that the assumption that a seller
prices uniformly all his sold units, is actually without loss of generality. So
our assumption that each seller prices all his units (not just sold units)
uniformly boils down to assuming that all unsold units are priced the same as
the sold units. Pricing uniformly is a natural assumption in many settings,
and often also a constraint of the marketplace. The second
assumption that each seller prefers selling the larger number of items implies
that among all utility equivalent prices, the seller picks the lowest price.

\paragraph{Last Considered Seller} For a given price vector $\prices$ where the
prices are uniform within a seller, the greedy procedure visits sellers in the
order of increasing prices breaking ties lexicographically. We call the
$\numSell$-th seller in this order as the {\em last considered seller}. Note that the greedy
algorithm could sometimes stop with {an earlier seller (as the marginal utility
becomes negative),}
% OLD: the first seller,
but the last considered seller is still defined to be the $\numSell$-th seller in the order.

Let $\moni$ denote the number of items that seller $i$ would sell to a buyer who
has already in possession of
$\supp - \suppi$ items, i.e., the number of items sold by
a monopolist seller to a buyer with valuation
$\val(\cdot|n - n_i)$ {(breaking ties towards larger quantity by our
assumption).
Formally, $\mu_i$ is the largest element in the set $\text{argmax}_k [ v(k \vert n - n_i) ]$.

The main step towards our main theorem (Theorem \ref{thm:main-epoa}) is to show
that after $\numSell \cdot v(1) /\disc$ rounds, at least $\min_i(\supp - \suppi
+\moni)$ are sold, {and the quantity sold will never drop below that amount in
later time-steps}.
We actually prove a stronger version of Theorem
\ref{thm:main-epoa} given by the following proposition:

\begin{proposition}\label{thm:soldGuarantee}
For every starting price $\prices^{0}$, and every possible order of
best-responses from sellers,
the number of units sold reaches $\min_i(\supp - \suppi +\moni)$ after at most
$\numSell\cdot v(1)/\epsilon$ rounds, and never falls below that amount. The
social welfare after
at most $\numSell \cdot v(1) /\disc$ rounds is at least\\
$\left[ 1+\log\left(\frac{\supp}{\min_i(\supp-\suppi+\moni)}\right) \right]^{-1}
\cdot \val(\supp)$
and never falls below.
\end{proposition}
The proof of the proposition follows directly from the series of lemmas below. 
%We prove Lemmas~\ref{lem:monopolistQuantity} and~\ref{lem:noFurtherFalling} in the full version of our paper~\cite{BPS14}.

\begin{lemma}\label{lem:monopolistQuantity}
When seller $i$'s best-response to prices $\pricesi[-i]$
makes him the last considered seller in the buyer's greedy procedure,
he sells exactly $\moni$ items (regardless of $\pricesi[-i]$).
\end{lemma}
\begin{proof}
We prove this by showing that whenever seller $i$'s best response to
$\pricesi[-i]$ makes him the last considered seller, his revenue is identical to
what a monopolist seller would get when dealing with a buyer with valuation
$\val(\cdot|n-n_i)$. This proves that the number of items sold by
the best-responding seller and the monopolistic seller is the same (the latter
is $\moni$ by definition) because we have assumed that sellers break tie towards
selling more items. To prove the claim, note that the monopolist seller could
not have made strictly smaller profit than the best responding seller $i$
because he could have just used the price of the latter and obtained as much
profit.  Similarly, the best-responding seller $i$ could not have made
strictly smaller profit than the monopolist seller $i$: the former could have
always used the latter's price. If this price resulted in the best-responding seller
$i$ being the last considered seller, he would have by definition of $\moni$
sold exactly $\moni$ items and made the monopolist's profit.  If this price
resulted in seller $i$ being considered before the last seller, he would have
sold at least $\moni$ items, and thus making at least as much profit.
\end{proof}

\begin{lemma}\label{lem:noFurtherFalling}
For every starting price $\prices^0$ and every possible order of best-responses
from sellers, the number of units sold reaches $\min_i(\supp - \suppi +\moni)$
within the first $\numSell \cdot v(1)/\disc$ rounds, and never falls below after
that.\end{lemma}
\begin{proof}
We prove this in two claims. First, we show that the number of units sold
reaches $\min_i(\supp - \suppi +\moni)$ within the first $\numSell\cdot
v(1)/\disc$
rounds, and second, we show that it doesn't fall below $\min_i(\supp - \suppi
+\moni)$.  In fact, we need only  $\numSell( v(1)/\disc-1)$ rounds to
reach $\min_i(\supp - \suppi
+\moni)$.

We begin with the first claim. Consider two cases.
\paragraph{Case 1} In the first $\numSell (v(1)/\disc-1)$ rounds, some
seller's best response
made him the last considered seller in the buyer's greedy procedure. The seller,
say $i$, who made this choice did so because the buyer at this price will
exhaust the supply of every other seller, and buy a few items from seller $i$
too. The exact number sold by seller $i$ will be $\moni$ as shown in
Lemma~\ref{lem:monopolistQuantity}, and the total number sold is $\supp - \suppi
+\moni$.

\paragraph{Case 2} In no time step in the first $\numSell(v(1)/\disc-1)$
rounds did
any seller best respond to become (or remain) the last considered seller in the
buyer's greedy procedure. Let $h^r$ be the largest price in the price vector at
the end of $rs$ rounds (so $h^0$ would the largest price in $\prices^0$). We
show by induction on $r$ that either $h^r \leq \max(h^0 - r\disc,\disc)$ or the
number of items sold already reached $\min_i (\supp - \suppi + \moni)$ at some
time step in the first $r$ rounds. Thus by the end $\numSell(v(1)/\disc-1)$
rounds, we are guaranteed to have sold at least $\min_i (\supp - \suppi +
\moni)$ units at some time step, proving the claim. The base case of $r=0$ is
trivial. Assume that for $r \leq k$ the number of items sold never reached
$\min_i (\supp - \suppi+\moni)$ in the first $r$ rounds, and that $h^r \leq
\max(h^0 - r\disc,\disc)$ for $r \leq k$.  Consider $r = k+1$. By the definition
of case 2, at no time step in any of the rounds between $ks+1$ to $ks+s$ (both
inclusive) did the largest price go strictly above $h^k$ as that would put this
in case 1.  Furthermore, in each of these $s$ rounds, some seller who is
currently the last considered seller in the buyer's greedy procedure has to
best-respond, and by definition of case 2, this seller has to strictly reduce
his price below $h^k$. Also, this seller will never again raise his price to
$h^k$ in these $s$ rounds as it will again violate case 2.  Thus in each of the
$s$ rounds between $ks+1$ and $ks+s$, as long as there is at least one seller
priced at $h^k$, at least one seller reduces his price by at least $\disc$.
Since there are $\numSell$ sellers, the number of sellers priced at $h^k$ is at
most $\numSell$, and thus after $\numSell$ rounds, the highest price would have
fallen down by at least $\disc$, thus proving the inductive step.

  We now prove the second claim, namely, the number of units sold doesn't fall
  below $\min_i (\supp - \suppi + \moni)$ once it reaches this quantity. At
  every time step $t$ after this quantity has been sold, the best-responding
  seller $i$'s action could either make (or retain) him the last considered
  seller in which case the number sold is $\supp - \suppi + \moni$ by
  Lemma~\ref{lem:monopolistQuantity}. If not, the best-responding seller $i$'s
  action will satisfy the condition of Lemma~\ref{lem:mostExpensive},
 % i.e., it
 % will not decrease the number of items priced at most the
 % price of the most expensive item currently getting
 % sold,
  in which case by Lemma~\ref{lem:mostExpensive} the number sold will not
  fall below what is currently getting sold.
\end{proof}

\begin{lemma}\label{lem:welfare-bound}
When the number of items sold is at least \\
$\min_i (\supp - \suppi + \moni)$, the
welfare is at least \\
$\left[ 1+\log\left(\frac{\supp}{\min_i(\supp -
\suppi + \moni)}\right)\right]^{-1} \cdot \val(\supp)$.
\end{lemma}
\begin{proof}
Let $t^*$ be the smallest time
step at which the number of units sold reaches $\min_i (\supp - \suppi +
\moni)$. Let $W^t$ denote the social welfare at the end of time-step $t$, and
let $\underline{W} = \min_{t\geq t^*} W^t$. Let $i^* \in \argmin_i
{(\supp-\suppi+\moni)}$. Note that
when $\supp - \suppi[i^*] + \moni[i^*]$ units are sold, by the decreasing marginals property,
the social welfare is at least
$(\supp-\suppi[i^*]+\moni[i^*]) \cdot m_{\supp-\suppi[i^*]+\moni[i^*]}$ where
$m_k$
denotes the marginal of the $k$-th unit for the buyer. If more items are sold, the welfare doesn't drop, and therefore we have
$\underline{W} \geq
(\supp-\suppi[i^*]+\moni[i^*]) \cdot m_{\supp-\suppi[i^*]+\moni[i^*]}$ .

By the definition of a monopolist who sells to a buyer with valuation
$\val(\cdot \vert n-n_i)$, we have that $\moni\cdot m_{\supp-\suppi+\moni} \geq
k\cdot m_{\supp-\suppi+k}$ for all $k\in\{1,2,\ldots, \suppi\}$. Therefore, by  the decreasing
marginals property, we have that for all $k$ such that $\mu_i \leq k \leq n_i$:
$$(\supp-\suppi+\moni)m_{\supp-\suppi+\moni} \geq
(\supp-\suppi+k)m_{\supp-\suppi+k} $$

Since at least $\supp - \suppi[i^*] +\moni[i^*]$ units are sold,
$\underline{W}$ will be short of $\val(\supp)$ by at most the welfare of the
remaining marginals. I.e., we have
$\val(\supp) - \underline{W} \leq
\sum_{k=\moni[i^*]+1}^{\suppi[i^*]}m_{\supp-\suppi[i^*]+k}$. Thus,
\begin{align*}
\val(\supp)-\underline{W}
& \leq
\sum_{k=\moni[i^*]+1}^{\suppi[i^*]}m_{\supp-\suppi[i^*]+k}\\
&\leq \sum_{k=\moni[i^*]+1}^{\suppi[i^*]}\frac{
(\supp-\suppi[i^*]+\moni[i^*])m_{\supp-\suppi[i^*]+\moni[i^*]}}{{\supp-\suppi[i^*]+k}}
 \\
&\leq
\log\left(\frac{\supp}{\supp-\suppi[i^*]+\moni[i^*]}\right)(\supp-\suppi[i^*]
+\moni[i^*])m_{\supp-\suppi[i^*]+\moni[i^*]}\\
&\leq
\log\left(\frac{\supp}{\supp-\suppi[i^*]+\moni[i^*]}\right)\underline{W}
\end{align*}
where the second inequality follows from the previous expression and {the third
from the inequality
$\sum_{j={r+1}}^t 1/j \leq \log(t/r)$ for any integers $t$ and $n$.} Rearranging
the terms, we obtain the desired result:
$$\underline{W} \geq
\frac{\val(n)}{1+\log\left(\frac{\supp}{\supp-\suppi[i^*]+\moni[i^*]}\right)} =
\frac{\val(n)}{1+\log\left(\frac{\supp}{\min_i(\supp-\suppi+\moni)}\right)}$$
\end{proof}

Lemmas~\ref{lem:noFurtherFalling} and~\ref{lem:welfare-bound} together
prove
Proposition~\ref{thm:soldGuarantee} and Theorem \ref{thm:main-epoa}. The
following corollary is an immediate consequence of
Theorem~\ref{thm:soldGuarantee} for the case where all sellers have the same
quantity to sell. The  corollary illustrates that if each
seller holds the same number of items, then the Eventual Welfare Guarantee
depends only on the number of sellers. Moreover, the equilibrium is efficient
in the limit when the number of sellers grows to infinity.

\begin{corollary}\label{cor:soldGuarantee-same}
If each seller holds the same number of items, then the Eventual Welfare Guarantee 
is bounded by $1+\log\left(1+\frac{1}{s-1}\right)$. In particular, it
tends to $1$ as $s \rightarrow \infty$.
\end{corollary}

%The second corollary illustrates that the price of anarchy is constant if no seller controls more than a constant fraction of the items in the market.

\begin{corollary}\label{cor:soldGuarantee-fraction}
If there is a constant $c \in (0,1)$, such that no seller controls more than a
$c$ fraction of items in the market, i.e., $n_i \leq c \cdot n$, then the
Eventual Welfare Guarantee is bounded by a constant: $1-\log (1-c)$. In
particular, it tends to $1$ as $c \rightarrow 0$.
\end{corollary}

\subsection{Lower bound on the Eventual Welfare Guarantee}

We next present a lower bound showing that the {efficiency bound}
% OLD: result  
of Theorem~\ref{thm:main-epoa} % {thm:soldGuarantee} 
is asymptotically  tight when $\frac{\supp}{\supp-\suppmax}$ grows large, i.e.,
the nearly-monopolistic setting.
\begin{theorem}\label{thm:lower_bound}
Fix any number of sellers $s$ and any positive supplies $\{\suppi\}_{i\in [s]}$ such that $\frac{\supp}{\supp-\suppmax}>3$.
There exists a decreasing marginal valuation for the buyer
such that for any sequence of best-responses from sellers,\footnote{and any
lexicographic order that determines the tie breaking rule}
the social welfare after any number of rounds is at most\\
$\left[ \frac{1}{5}\log\left(\frac{\supp}{\supp-\suppmax}
\right)-1\right]^{-1} \cdot \val(\supp)$.
\end{theorem}
\begin{proof}
Recall that $\suppmax=\max_i\ \suppi$ and $\supp=\sum_i \suppi$.
To prove the claim we need to show that % for a small enough $c>0$
the ratio of
$v(\supp)$ to the welfare in every round is at least
$\frac{1}{5}\left(\log\left(\frac{\supp}{\supp-\suppmax}\right)-1\right)$.

% If $n_{max}\leq 2n/3$ then by Theorem~\ref{thm:soldGuarantee} eventually at least a third of the items will be sold and the efficiency ratio will be at least third, and thus for the claim to hold it is sufficient to check that $\frac{1}{3}\geq \frac{1}{5}\left(\log\left(3\right)-1\right)$ (since $3\geq \frac{\supp}{\supp-\suppmax}$), and this indeed holds. % holds for $c\leq 1/3$.

% We next 
Recall that we have assumed that $\frac{\supp}{\supp-\suppmax}>3$ and thus $\suppmax> 2\supp/3$.
Let $r=\supp-\suppmax< n/3$ be the total supply of all sellers but the one with $\suppmax$ units.
Assume that $\disc<1$ is small enough.
Consider the following valuation of the buyer.
He values each of the first $2r$ units for $1+\disc$ each, he values the next
$r$ items for $1/3$ each, the next $r$ for $1/4$ each etc.,
till we reach $n$ units (note that the last bunch might have less than $r$ units).
%Assume that the seller with $\suppmax$ units is first and the other are orderer in some arbitrary order after him. Recall that ties are broken lexicographically. Also consider a round robin order of moves according to the lexicographic order.
We claim that at most $2r$ items will ever be sold. Indeed, if
%more then
$\ell > 2r$ items are sold, then the seller with $\suppmax$ units must be
selling at least $\ell-r>r$ units (as all other sellers combined supply only $r$ units).
To sell such a quantity he must be pricing his items at price at most $\lceil (\ell-r) /r \rceil^{-1}$, since this
is the $(\ell-r)$-th marginal of the buyer, therefore, getting revenue at most
$(\ell-r) / \lceil (\ell-r) /r \rceil \leq r$. However, by posting price
$1+\epsilon$, this seller is guaranteed revenue at least $(1+\epsilon) r$.

As at most $2r$ items are ever sold, the welfare at every round is at most $2r(1+\disc)$, while
the welfare $v(\supp)$ is at least $r(2+\sum_{j=3}^{\lfloor n/r\rfloor}
\frac{1}{j})\geq r(\frac{1}{2}+\sum_{j=1}^{\lfloor \supp/r\rfloor}
\frac{1}{j})\geq r(\frac{1}{2}+ \log ({\lfloor \supp/r\rfloor}))
\geq r\left(\log (\supp/r)-  \frac{1}{2}\right)$ since $\log ({\lfloor
\supp/r\rfloor})\geq \log ({\supp/(2r)})= \log ({\supp/r})-1$ because
$\supp>2r$.
We conclude that the ratio of $v(\supp)$ to the welfare in every round is at least\\
$\frac{\log (\supp/r)- \frac{1}{2}}{2(1+\disc)} \geq \frac{1}{4(1+\disc)}
\left(\log \left(\frac{\supp}{\supp-n_{max}}\right) - 1 \right)$\\
$\geq
\frac{1}{5} \left(\log
\left(\frac{\supp}{\supp-n_{max}}\right) - 1 \right)$.
\end{proof}

\ncomment{
In order to see that the number of rounds to reach good welfare state is tight,
consider the following example:

\begin{example}
 Consider $s$ sellers, where $s-1$ sellers hold $1$ item and one seller
holds $n-s+1$ items. Also, consider the buyer as having marginal valuation $2
\beta $ for the first $s-1$ items and marginals $\frac{\beta}{s},
\frac{\beta}{s+1}, \hdots, \frac{\beta}{n}$ for the remaining items. Assume that
initially all sellers price their items uniformly at $2 \beta$. Schedule the
best response sequence such that in each round, the \emph{last considered
seller} is the last to best-respond.

We note two things: if the smallest price posted is larger then $\beta$, no
seller will deviate to post less then the smallest price minus $\epsilon$, not
even the seller controlling most of the items.

Now, after $r$ rounds, the
price is $\lceil 2 \beta- \frac{r}{s} \epsilon \rceil$, since
\end{example}

}

\section{Pricing Game with Uncertain Demand for Homogeneous Goods}
\label{sec:bayesian}
In this section we drop the assumption that the valuation of the buyer is common
knowledge among the sellers.
Instead, while we still assume that the buyer has an intrinsic valuation $v^*
\in \V$,
we no longer assume that it is known to the sellers.
Rather, we assume that each seller has a prior belief $\B_i^0$,
expressed as a probability distribution over $\V$.
Note that the initial beliefs of different sellers might be different.
At each step, sellers observe the quantity sold and update their beliefs in a
Bayesian way.

We remind the reader that $\V$ is a set of functions mapping quantities to
values in  $\epsilon \Z_+$,  and
thus $\V$ is a countable infinite set.
In our main theorem for the pricing game with uncertain demand present below,
we assume that the prior belief of each seller has finite
support, i.e., it places positive probability on finitely many valuation
functions.
A natural example of a finite set of beliefs is the one where all sellers have
an upper bound $M$ on the total value of the buyer $\val(\supp)$, that is,
$\B_i^0(v) = 0$ for all $v$ such that $\val(\supp) > M$. In this case the
beliefs form a finite set.

\begin{theorem}\label{thm:stochastic-demand}
If the prior belief of each seller has finite support and all priors are
consistent\footnote{i.e., they place positive probability on the true valuation of
the buyer}, then the Eventual Welfare Guarantee of the uncertain demand game is
at most $\alpha = \frac{\supp}{\supp-\suppmax+1}$.
Moreover, in any best response sequence, the welfare is lower than an $1/\alpha$ fraction of
the optimal welfare in at most $\frac{s \cdot K}{\epsilon} \cdot \sum_i
\abs{\text{supp}(\B_i^0)}$ rounds, where $\text{supp}(\B_i^0) = \{v \in \V;
\B_i^0(v) > 0 \}$ and $K = \max_{v \in
\cup_i \text{supp}(\B_i^0)} v(1)$. If all sellers start with the same
initial belief $B^0$, the number of rounds will decrease to $\frac{s \cdot
K}{\epsilon} \cdot \abs{\text{supp}(\B^0)}$.
\end{theorem}

We prove Theorem~\ref{thm:stochastic-demand} in Appendix~\ref{app:bayesian}.

We note two main differences between Theorem
\ref{thm:stochastic-demand} and its counterpart in the full information game
(Theorem \ref{thm:main-epoa}). The first is the welfare ratio: in the
nearly-monopolsitic setting in which one seller controls a large fraction of the
market, i.e. $\frac{\supp}{\supp-\suppmax+1} \rightarrow \infty$, the bound is
exponentially worse than the full information bound.
{We will show in Theorem~\ref{thm:stochastic-demand-lb} that welfare can
indeed be very low if beliefs do not converge to the true valuation.}
This difference highlights that the efficiency of the market can exponentially
improve if sellers can accurately learn the valuation of the buyer.
For the other extreme case where no seller controls a significant
fraction of the market, i.e. $\suppmax = \epsilon n$, the two bounds become
approximately $1+\epsilon$ as $\epsilon \rightarrow 0$.

The second difference is that although both theorems guarantee that best response dynamics will lead
to good welfare states in finite time, the first bound guarantees
that except for possibly the first $t' = \frac{s \cdot v(1)}{\epsilon}$ rounds,
the welfare will be high. Theorem \ref{thm:stochastic-demand} gives a much
weaker guarantee: it shows that welfare will be good for all but at most $t'' =
\frac{s \cdot K}{\epsilon} \cdot \sum_i \abs{\text{supp}(\B_i^0)}$ rounds. Yet
it doesn't guarantee those will be the first  rounds. We leave the question of
getting a stronger convergence guarantee for this setting as an open problem.

We next present corollaries for the uncertain demand setting that are
parallel to Corollaries \ref{cor:soldGuarantee-same} and
\ref{cor:soldGuarantee-fraction}.

\begin{corollary}
If every seller controls the same number of items, then the Eventual Welfare Guarantee is
{at most $1+\frac{1}{s-1}$.}
In particular, it tends to $1$ as $s \rightarrow \infty$.
\end{corollary}

\begin{corollary}
If every seller controls {at most} a $c$ fraction of the market, then the
Eventual Welfare Guarantee is
{at most $\frac{1}{1-c}$}.
In particular, it tends to $1$ as $c \rightarrow 0$.
\end{corollary}

And finally, we show that the bound we have presented for The Eventual Welfare Guarantee of uncertain demand games is essentially tight. We show that it can be as large as $\Theta(n)$ even when the beliefs of all sellers are identical. The proof of this Theorem (Theorem~\ref{thm:stochastic-demand-lb}) follows directly from the example below. 

\begin{theorem}\label{thm:stochastic-demand-lb}
The Eventual Welfare Guarantee of an uncertain demand game can be as large as $\Theta(n)$.
\end{theorem}

\begin{example}
\label{ex:stochastic-lb}
Consider a buyer with a valuation $v^*$ over $n$ items with marginals
$1+\epsilon, 1+\epsilon, \frac{1}{3}, \frac{1}{4},\frac{1}{4},\frac{1}{4},
\hdots, \frac{1}{4}$. Now, let $v$ be the valuation function with marginals
$1+\epsilon, 1+\epsilon, \frac{1}{3}, \frac{1}{4}, \frac{1}{5}, \frac{1}{6},
\hdots, \frac{1}{n}$. Consider also two sellers holding $n_1 = 1$ and $n_2 =
n-1$ items and with beliefs
{$\B_1^0= \B_2^0$} assigning probability $\delta$ to $v^*$ and
$1-\delta$ to $v$.
If $\delta$ is small enough, both sellers will best respond
assuming that the buyer chooses how many items to buy according to $v$. Since
for this best-response sequence  more than $3$ items will never be sold (since
the situation here is the same as in the example used in the proof of Theorem
\ref{thm:lower_bound}), the sellers will never be able to refine their beliefs.
This leads to an best-response dynamics in which the welfare is {less than $3$}
for each time-step, while the optimal welfare is $\frac{n}{4}$.
\end{example}

\section{Heterogeneous Goods}
\label{sec:hetero}
As explained in the introduction, the pricing game with heterogeneous goods has very different properties from that for homogeneous goods. A good place to start studying heterogeneous goods is the case of two extremes: a buyer with additive valuations, and a buyer with unit-demand valuations. In both these cases, it turns out that an efficient pure Nash equilibrium is guaranteed to exist.
\begin{fact}
\label{fact:additive}
 The full information pricing game with sellers holding heterogeneous goods, and the buyer having additive valuation, is guaranteed to have an efficient pure Nash equilibrium.
\end{fact}
\begin{fact}
\label{fact:unit-demand}
 The full information pricing game with sellers holding heterogeneous goods, and the buyer having unit-demand valuation, is guaranteed to have an efficient pure Nash equilibrium.
\end{fact}

We prove these facts in Appendix~\ref{app:hetero}. Given that the two extremes are well behaved, we move to explore $k$-additive valuations, the simplest generalization of the two extremes.
We already see that there need not exist efficient equilibria, even when an equilibrium exists.

\begin{fact}\label{fact:additiveNoEff}
The full information pricing game with just two sellers, and a $3$-additive buyer need not have an efficient pure Nash equilibrium.
There is such a setting in which equilibrium exists but is not efficient.\footnote{The equilibrium is unique up to price of the unsold item, which can be increased with no change to the utilities of both sellers.}
\end{fact}
\begin{proof}
This fact is again easy to see. Consider two sellers with seller $1$ holding $1$ item of value $\epsilon$ to the buyer and seller $2$ holding $3$ items of value $2\epsilon$ each to the buyer. The buyer breaks ties towards seller $1$. If seller $1$ prices his item at $\epsilon$, then for seller $2$, the best response is to price all the items at $2\epsilon$. This pricing pair is an equilibrium, and it is easy to verify that in any equilibrium the two sellers price the sold items the same and have the same utility as in this equilibrium.
This results in a buyer welfare of $5\epsilon$, which is strictly smaller than the optimal welfare of $6\epsilon$.
\end{proof}

An interesting fact about the example in the proof of Fact~\ref{fact:additiveNoEff} is that if instead of $2$ sellers, all the $4$ items were held by a single monopolist seller, the social welfare would have been at its optimal level ($6\epsilon$),
strictly larger than with two competing sellers.
While one might intuitively think that splitting the items between two sellers could only increase welfare as it increases competition,
this example shows that it could also {\em decrease welfare in equilibrium} (note that equilibrium welfare is well defined when all the sellers' and buyer's utilities are same across all equilibria, like in the example used in the proof of Fact~\ref{fact:additiveNoEff}).
\begin{remark}
When the items held by a monopolist are split between two sellers, equilibrium welfare could decrease.
\end{remark}

The proof of Fact~\ref{fact:additiveNoEff} presents an example with equilibrium welfare loss of $\epsilon$.
Can the loss be much larger?
In Theorem~\ref{thm:k-additive} we show that with a $k$-additive buyer, whenever there exists an equilibrium, the loss in welfare cannot get larger than $2k\epsilon$.
Thus, the loss in welfare vanishes with $\epsilon$, and the welfare {\em in equilibrium} is almost optimal.
Yet, unlike in the unit-demand and additive cases, for $k$-additive valuations equilibrium might fail to exist and thus the welfare guarantee is much weaker as it relies on assuming that an equilibrium is reached, but for some valuations that might not be the case.
What happens when equilibrium is not reached? Is welfare still guaranteed to be high?
We next show that in some settings, welfare can be very low.
% Actually, in settings in which equilibrium is not reached, the welfare can be terribly

\subsection{Out of equilibrium cycling}
%\paragraph{Out of equilibrium cycling}
%It turns out that for heterogeneous items an equilibrium does not necessarily exist, even with a $k$-additive buyer.
We show that even for a $k$-additive buyer, when equilibria don't exist, welfare can be very poor in a strong sense. First, the Eventual Welfare Guarantee can be arbitrarily large. Further, only in $\frac{1}{s}$ fraction of the time does the welfare reach even a non-zero fraction of the optimal welfare (and thus, the difference between the welfare most of the time, and the optimal welfare, is very large).

\begin{theorem}
The full information pricing game with heterogeneous goods, and a $k$-additive buyer does not always have a pure Nash equilibrium. When a pure Nash equilibrium does not exist, the Eventual Welfare Guarantee is $\infty$. Further, there exist best response sequences where the welfare obtained approximates the optimal welfare within a non-zero factor only for $O(\frac{1}{s})$ fraction of the time.
\end{theorem}
\begin{proof}
Consider $s$ sellers with seller $1$ holding one item item of a large value $V$, and sellers $2$ to $s-1$ holding one item each of a small value, say $10$. Seller $s$ holds $s = \frac{1}{\epsilon}$ items of value $10$ each.  The buyer is $s$-additive, and prefers sellers in the order $1,\dots,s$. Sellers best respond in the order $s,\dots, 1$. Let the price sequence begin from all  prices being $\infty$. First, seller $s$ sets a price of $10$ for each of his items and gets the maximal possible revenue of $10s$ for his items. Then seller $s-1$ updates his price from $\infty$ to $10$, following by seller $s-2$ and so on till seller $2$ updates his price to $10$. Then, seller $1$ will update his price from $\infty$ to $V$. In the second best response cycle, seller $s$ slashes his price down to $10-\epsilon$ each, and so will sellers $s-1$ to $2$ in sequence, and then seller $1$ will slash his price down to $V-\epsilon$. This price competition will continue until seller $s$ slashes his prices down to $10\epsilon$ each, and so do sellers $s-1$ to $2$, and seller $1$ sets his price at $V-10+ 10\epsilon$. At this point, seller $s$ will no more drop his price, but raise the prices of all his items to $10$ each, upon which sellers $s-1$ to $2$ will update their price to $10$, and seller $1$ will update his price to $V$, bringing the sequence back to where it started. Thus in this infinite looping, only at the stage immediately after seller $1$ responded does the item of value $V$ ever get bought. At other times, only items of value $10$ get bought leading to social welfare of $10s$. Since $V$ can be made arbitrarily big compared to $s$, a non-zero approximation to welfare is attained only when seller $1$ responds, which happens $O(\frac{1}{s})$ fraction of the time.

To see that no pure equilibrium exists in this setting, let $p$ be the price of the most expensive item sold by seller $s$ in some equilibrium. Note that $p > 0$ because seller $s$ is guaranteed to get a revenue of at least $10$ by just selling one item priced at $10$. Given this, sellers $2,\dots,s-1$ best respond by setting prices of $p$ each, and $1$ sets a price of $V-10+p$. This would mean that in any equilibrium the buyer buys exactly one item from seller $s$, and thus $s$ has to price it at $p=10$, and sellers $2,\dots,s-1$ best respond with a price of $10$, and $1$ sets a price of $V$. But given this, seller $s$ is not best responding with a price of $10$: he can sell all his items at a price of $10-\epsilon$ each.
\end{proof}

\subsection{In equilibrium bounds for a $k$-additive buyer}
We next show that with a $k$-additive buyer, whenever there exists an equilibrium, the loss in welfare cannot get larger than $2k\epsilon$. This is almost tight: we show an example where there can be an additive loss of $k\epsilon$. As a corollary, we get that whenever all items are of value at least $r\epsilon$ with $r > 2$, the Price of Anarchy in a $k$-additive setting is at most $\frac{r}{r-2}=1+\frac{2}{r-2}$, which is bounded from above by $3$ and for any fixed valuation decreases monotonically to $1$ as $\epsilon$ goes to $0$ (which implies that $r$ grows to infinity).

\begin{theorem}
\label{thm:k-additive}
Every pure Nash equilibrium in a full information pricing game with heterogeneous goods and a $k$-additive buyer obtains a welfare of at least optimal welfare less $2k\epsilon$. %$2(k-1)\epsilon$.
\end{theorem}
We prove the theorem in Appendix~\ref{app:hetero}.
The lost welfare of $2k\epsilon$ is tight up to constant factors. Here is an example where losing welfare of magnitude $k\epsilon$ is inevitable.
\begin{example}
Consider two sellers. The first holds $k$ items of value $0$ each, and the second holds $k$ items of value $\epsilon$ each. Suppose the first seller is more preferred. Then, the first seller selling all his items at $0$, and the second seller posting a price of $\epsilon$ for all his items is an equilibrium with a welfare of $0$, causing an additive loss in welfare of $k\epsilon$.
\end{example}

The example shows that the Price of Anarchy can be unbounded. Yet, this happens when the value of items is very low (only $\epsilon$).
We next show that this cannot happen when the value is large relative to $\epsilon$, in all such settings the PoA is small.

\begin{corollary}
For $r > 2$, if all items are of value at least $r\epsilon$, the pure Nash price of anarchy of the full information pricing game with heterogeneous goods and a $k$-additive buyer is at most $\frac{r}{r-2}$. As $r$ increases, the PoA monotonically decreases, approaching $1$.
\end{corollary}

\section{Conclusion and Open Problems}
\label{sec:conc}
In this paper we consider a model of price competition that exhibits price
fluctuations and we provide welfare guarantees, despite the lack
of convergence to equilibrium. We did so in a very simple model of seller
behavior: myopic best response. It would be interesting to try to do the same
for more sophisticated models of {behavior},
such as Sequential Equilibrium
\cite{kreps_wilson} or Markov Perfect Equilibrium \cite{markov_perfect_eq}.
{In such models a seller would reason not only about his utility in the
next round but also about the influence of his action}
on the long-run utility --
measured, say, by time discounted payoffs. For the uncertain demand setting,
sellers would also face the trade-off between learning more about the buyer
valuation and maximizing gains given their current information.

Our results also highlight the value of accurately knowing the valuation
function of the buyer. In the setting with uncertainty, the Eventual Welfare Guarantee can be exponentially worse than in the full information setting. Are
there natural learning strategies that sellers could employ that could overcome
this gap and at the same time preserve some sort of equilibrium behavior ?

Another direction to extend in the setting of heterogeneous items is to understand what buyer valuations yield good Eventual Welfare Guarantee. An important ingredient in our analysis is that we have a precise description of how a buyer reacts to price updates. For heterogeneous items, understanding the effect of price updates is considerably more challenging.

%Other interesting direction is to extend our results to heterogeneous items. If
%the items are not identical but the buyer's valuation is such that he still
%chooses which items to buy via a greedy algorithm (e.g., gross substitute
%valuations), is it possible to get similar results ? An important ingredient in
%our analysis is that we have a precise description of how a buyer reacts to
%price updates. For heterogeneous items, understanding the effect of price
%updates is considerably more challenging.

One other feature absent from our model is the intrinsic costs of sellers. For
example, in the retail market for gasoline the cost is typically the wholesale
price. In the case of publishers and impressions, the cost of publishers to
serve ads corresponds to the dis-utility of the website-users in seeing the ads.
{It would be interesting to see how non-zero costs would
alter the model.
%Finally, one could consider fully Bayesian settings in which there is also Bayesian uncertainty about other sellers' costs.}

%\section{TO DO}
%\input{todo.tex}

%\bibliographystyle{plain}
\bibliographystyle{acmsmall}
\bibliography{sink}

\appendix
\section{Random Tie Breaking}
\label{app:randomTie}
We next show that with random tie breaking, as long as $\disc$ is small enough,
no best response will result in a tie. One can use this to derive similar bounds as we do in the paper.

A random tie breaking is such that items of the same price are considered in a uniformly random order by the greedy algorithm. 

\begin{proposition}
Consider the full information game. 
Assume that the buyer has a decreasing marginal valuation with
$m_n=\val(n)-\val(n-1) > 0$ and there are at least $2$ sellers. 
In the full information game with the random tie breaking rule and small enough $\disc$, any best response of a seller will never result with an item of this seller that is priced at exactly the same price as an item offered by another seller, unless this item is sold with probability 0 or 1.
\end{proposition}
\begin{proof}
Consider such a best response in which the item is priced at $p$ and sold with probability $z>0$, $z<1$.
We argue that this seller can strictly increase his utility by decreasing the price of that item by $\disc$,
as long as $\disc$ is small enough.
% Assume that $\disc<m_n/n^3$. If $p<m_n/n$

Observe that before the change in price, any item with price smaller than $p$ is sold, and any item with price larger than $p$ is not sold. 
This is still true after the price decrease, and in particular, the item for which the price has decreased is sold with probability 1. 
Assume that there is a total of $k>1$ items priced at $p$ before the price decrease, out of which $m<k$ are from the considered seller. 
Assume that $r<k$ items are sold at price $p$ before the price decrease. 
The probability of each of the items priced at $p$ to be sold before the price change is $z=r/k$, and after the price change the probability each of the items priced at $p$ to be sold is $(r-1)/(k-1)$, and now the seller is pricing only $m-1$ items at this price.

The increase in utility for this seller is:
$$(p-\disc)+p\left(\frac{(m-1)(r-1)}{k-1}- \frac{m r}{k}\right)=
p\left(\frac{(k-m)(k-r)}{k(k-1)}\right)-\disc
%\geq p\left(\frac{1}{k(k-1)}\right)-\disc
$$
which is positive for small enough $\disc$ as long as $p$ is no too small, but $p\geq m_n$ since otherwise $z=1$.
\end{proof}

\section{Nonexistence of Pure Nash Equilibria in Full Information Game for Homogeneous Goods}
\label{app:fullInfo}
We now show that for the example~\ref{example:no-Nash-eq} presented in Section~\ref{sec:fullInfo} there does not exist a pure Nash equilibrium. 

\begin{oneshot}{Example~\ref{example:no-Nash-eq}}
Consider a setting with two sellers, each holding 2 units of the same good.
The buyer has diminishing marginal valuations, of $5,5,3,1$, i.e.,
$\val(1) = 5$, $\val(2) = 10$, $\val(3) = 13$, $\val(4)=14$. This setting does not have any pure Nash equilibrium.
\end{oneshot}

Clearly, the number of units sold at any pure NE cannot be 1 because every
seller can always make a non-zero revenue. It cannot be two either. There are
two ways in which two goods can be sold. Both goods sold from the same seller,
which is not possible at any equilibrium. The second way is to have just one
good sold from each seller. This can happen only if the unsold item from each
seller is priced above $3$. Either seller can price the unsold item at $3$ (and
the sold item also at $3$) to get a revenue of $6$ which is larger than $5$
which is the maximum they would have gotten previously. Next, the number sold at
pure NE cannot be three too. Without loss of generality, we assume that this
was reached with seller 1 selling two items and seller 2 selling one item. Let
$p_1$ and $p_2$ denote the prices used by seller 1. Given this, seller 2 can
obtain a revenue of at least $p_1 + p_2 -2\epsilon$ by pricing at $p_1-\epsilon$
and $p_2-\epsilon$. If $p_1 + p_2 > 3+2\epsilon$, then seller 2 now gets a
revenue strictly larger than $3$ which is the maximum he could have gotten by
selling one item when three items were totally sold. On the other hand, if
$p_1+p_2 \leq 3+2\epsilon$, seller 1 could have priced his items at $3$ and $1$
and gotten a total revenue of $4 > 3+2\epsilon = p_1+p_2$. Finally, the number
of goods sold at any pure NE cannot be 4 because that would imply that all the
four units were priced at $1$, and one of the sellers can instead price both his
items at $3$, sell one item and get a larger revenue of $3$ instead of $2$.

\section{Pricing Game with Uncertain Demand and Homogeneous Goods}
\label{app:bayesian}
\paragraph{Uniform pricing}
In the uncertain demand setting it is no longer true that a seller
always has a best response that sets uniform price for all his units. Consider for example a
monopolist seller whose belief about the valuation of the buyer is as follows:
he believes that with probability $.99$ the seller has marginals $1,1,0$ and
with probability $.01$ he has marginals $10,10,10$. Then, his unique best
response is to set prices $1,1,10$ for his items. However, we can argue that if
a seller has no uniform price best-response, after posting these prices % by playing this
he will necessarily learn something about the buyer's valuation.

At any given time-step $t$, given the set of beliefs $\B_1^{t-1}, \hdots,
\B_s^{t-1}$ formed by the sellers in the previous time-step, we say that a price
vector $\prices$ is \emph{informative} if for at least one buyer, there are
$v,v' \in \text{supp}(\B_i^{t-1})$ such that $\abs{X^v(\prices)} \neq
\abs{X^{v'}(\prices)}$. Otherwise we say that the price vector is
\emph{uninformative}.
We say that seller $i$ has an {\em uninformative best response} to
$\pricesi[-i]$ if $i$ has a best response $\pricesi$ such that the price vector
$(\pricesi[i], \pricesi[-i])$ is uninformative.

We observe now that if $\prices^1, \prices^2, \hdots$ is a best-response
sequence, there can be at most $\sum_{i=1}^s
\abs{\text{supp}(\B_i^0)-1}$ informative prices, since if $\prices^t$ is
informative, then $\sum_{i=1}^s \abs{\text{supp}(\B_i^t)} < \sum_{i=1}^s
\abs{\text{supp}(\B_i^{t-1})}$ and once $\abs{\text{supp}(\B_i^t)} = 1$ for all
$i$, then all prices are uninformative. We assume that the seller
favors uniform price uninformative best responses over non-uniform price
uninformative ones. In other words, when an uninformative
uniform best response is available, the seller will never play a
non-uniform best response unless it is informative.
{The following lemma shows that it is indeed the case that when a seller has an
uninformative best response, he has an uniform one as well.}

\begin{lemma}\label{lemma:uniform_uninformative}
If a seller has an uninformative best-response, then he also has a uniform price
best-response. Therefore, if sellers break ties in favor of uniform prices,
there will be at most $\sum_{i=1}^s \abs{\text{supp}(\B_i^t)} - s$
time-steps in which it does not hold that every seller will post uniform prices
across his items.
\end{lemma}

\begin{proof}
If a seller has an uninformative best-response $\prices_i$, then for every $v
\in \text{supp}(\B_i^{t-1})$, $X^v(\prices_i, \prices^{t-1}_{-i})$ is the same.
In particular, the most expensive item bought has the same price for
every $v$. So we can use the exact same argument as in Lemma
\ref{lem:unifPricing} to guarantee that there is an uniform best-response.
\end{proof}

\begin{lemma}\label{lemma:best_response_eventual}
Consider a (possibly infinite) interval in a best-response sequence where the
price posted in each time-step is uninformative.
{Then in every time-step in the interval that is after the first $\frac{K
\cdot s}{\epsilon}$ rounds, it holds that at least $n-n_{\max}+1$ items will be
sold.}
\end{lemma}

\begin{proof}
We note that Lemma \ref{lem:mostExpensive} is a statement about the buyer
behavior, which is the same in this setting -- so the lemma continues to hold.
Lemma \ref{lemma:uniform_uninformative} guarantees that the conclusion of Lemma
\ref{lem:unifPricing} holds in the uncertain demand setting if the prices being
posted are uninformative. Finally, the exact same argument used in
Lemma \ref{lem:noFurtherFalling} can be replayed to show that in less than
$\frac{s \cdot K}{\epsilon}$ rounds, {there will be a time-step in which} all
sellers except one will sell their
entire supply, making the total amount sold at least $n-n_{i}+\tilde{\mu}_i$ and
from that point on this amount will be sold for all subsequent time-steps. Where
$\tilde{\mu}_i\geq 1$
is the amount that a monopolist seller would sell at the
optimal price when facing a buyer with valuation $v(\cdot \vert n -
n_i)$ drawn from the distribution representing the limit belief of seller $i$.
Note that since all prices from the interval are uninformative, the belief of
the seller is constant inside that interval.
\end{proof}

We observe that the limit belief of each seller might not be accurate, in which
case $\tilde{\mu}_i$ might be far from the value of {{$\mu_i$ taken with respect
to the true valuation $v^*$}
(as we demonstrate in Example \ref{ex:stochastic-lb}). The fact that the
last considered seller doesn't best respond properly is the main reason why the
bound is linear rather than logarithmic.

\begin{proofof}{Theorem \ref{thm:stochastic-demand}}
The previous lemma guarantees that in all but
 {$\frac{s \cdot K}{\epsilon} \cdot \sum_i \abs{\text{supp}(\B_i^0)}$ rounds,}
$\supp-\suppmax+1$ items are sold.
Since $v$ has the decreasing marginals property,
then for every time-step $t$ that is not one of these rounds
$W^t \geq v(n-n_{\max}+1) \geq
\frac{n-n_{\max}+1}{n} \cdot v(n)$.
This completes the proof of the theorem.
\end{proofof}

{Theorem \ref{thm:stochastic-demand} is stated for the case that sellers
are completely myopic and do not post non-optimal prices only for the sake of
learning the valuation of the buyer (exploring).
We remark that the same proof (and thus the same result) also holds if sellers
{\em do} try to learn the valuation of the buyer by sometimes posting
informative prices that are not best responses for them. This is true as long as
it holds that in every time-step in which a seller is not exploring  and in which
prices are uninformative, the seller posts a uniform-price best response (which
we know that is available to him).}

\section{Heterogeneous Goods}
\label{app:hetero}
\begin{oneshot}{Fact~\ref{fact:additive}}
 The full information pricing game with sellers holding heterogeneous goods, and the buyer having additive valuation, is guaranteed to have an efficient pure Nash equilibrium.
\end{oneshot}
\begin{oneshot}{Fact~\ref{fact:unit-demand}}
 The full information pricing game with sellers holding heterogeneous goods, and the buyer having unit-demand valuation, is guaranteed to have an efficient pure Nash equilibrium.
\end{oneshot}
\begin{proofof}{Facts~\ref{fact:additive} and~\ref{fact:unit-demand}}
These facts are easy to see. In the additive valuation case, it is an equilibrium for each seller to price every good he owns exactly at the buyer's value for that good. An additive buyer breaking ties towards buying items of zero marginal utility will buy the entire set of items, getting the optimal social welfare. In the unit-demand case, if there is only one type of good, clearly all sellers pricing it at $0$ is an efficient pure Nash equilibrium. Suppose there were two or more types of good, let $v_{max}$ the buyer's value for his most preferred item type and let $i^*$ denote the most preferred seller (lexicographically least) holding such a good. Let $v_2$ be the most valuable good owned by some seller other than $i^*$. Then $i^*$ pricing all his items at $v_{max} - v_2$, and the other sellers pricing all their items at $0$ is an efficient equilibrium.
\end{proofof}

\begin{oneshot}{Theorem~\ref{thm:k-additive}}
Every pure Nash equilibrium in a full information pricing game with heterogeneous goods and a $k$-additive buyer obtains a welfare of at least optimal welfare less $2k\epsilon$. %$2(k-1)\epsilon$.
\end{oneshot}
\begin{proof}
Let $W^*$ be a welfare maximizing set of items defined by mimicking the buyer's greedy algorithm and his tie breaking procedure. Start with an empty set, and recursively pick the item of maximum value till the set is of size $k$ or all items get exhausted. When multiple items are of same value, pick the item from the most preferred seller, and when multiple items from same seller are of same value, break ties arbitrarily. This arbitrary tie-breaking rule makes $W^*$ non-unique. While this non-uniqueness is straight-forward to handle\footnote{For instance, note that all $W^*$'s that satisfy this definition, not only get the same welfare, but also contain the same number of items from every seller, and contain the same value profile of items from each seller. }, for clarity of proof, we fix some arbitrary tie breaking rule and assume that $W^*$ is uniquely defined.

Let $I_i$ be the set of items owned by seller $i$, and let $I = \cup_i I_i$. Fix an arbitrary pure strategy equilibrium and let $S_i$ be the set of items sold by seller $i$ in this equilibrium. Let $S = \cup_i S_i$.
Note that $|S| = \min(k,|I|)$, i.e., when $k$ or more items are available, no equilibrium will involve the buyer buying less than $k$ items. Let $\underline{v}$ be the value of the least valued item that gets sold.
%\footnote{Assuming the value of at least $k$ items is positive, which holds without loss of generality as otherwise the result for an additive buyer (Fact~\ref{fact:additive}) implies the theorem.}
 %Thus $W^*\cap S_i$ is the set of items sold by seller $i$ from $W^*$.

We prove this theorem by proving several claims in sequence. All of them will refer to the order in which the buyer approached sellers according to his greedy procedure on utilities. Thus it is useful to keep in mind the global order of utilities offered by items in equilibrium (where the ordering is defined with ties broken as defined before). Let $1, 2,\dots , s$ denote the lexicographical ordering of sellers.
% MB 10.21.14: added below. Is this correct?
% The first claim below shows that each seller will be approached exactly once by the greedy algorithm, thus one of them will be approached last.
Let $\ell$ be the seller who was approached last by the greedy procedure, and, sold at least one item. Note that $\ell$ is not necessarily equal to $s$.
%{\bf MOSHE TO BALU: as we discussed, $\ell$ is the last selling seller, not the last considered seller.}
\begin{claim}
\label{cl:equalUtilityLocal}
In any pure equilibrium, all items sold by a seller offer equal utility for the buyer.
\end{claim}
\begin{proof}
Suppose not: consider an arbitrary seller whose sold items offered different utilities for the buyer. The seller could consider reducing the utilities of all items to match that of the sold item offering the lowest utility. The fact that  the buyer, among items of same utility, picks from the lexicographically most favorable seller, and the fact that the lowest utility item was previously sold, together imply that the seller continues to sell the same set of items as before, and gets a higher revenue. This contradicts that the seller is best-responding.
\end{proof}
\begin{corollary}
In any pure equilibrium, all items sold by a seller are sold in one contiguous stretch.
\end{corollary}
The corollary implies that each seller is approached at most once by the algorithm, so we can talk about the order that sellers are approached (also note that every seller that is approached, sells at least one item).

\begin{claim}
\label{cl:losingWelfareEll}
In any pure equilibrium, all sellers $i\neq \ell$ sell all items from $W^* \cap I_i$ that have values at least $\underline{v} + 2\epsilon$.
\end{claim}
\begin{proof}
Suppose there is a seller $i$ with unsold items of value at least $\underline{v} + 2\epsilon$. In this case, seller $i$ could consider dropping the price of such an unsold item in steps of $\epsilon$. This process will increase the utility of that item, until it matches that of a sold item offering least utility. Due to the ``$2\epsilon$'' assumption in the lemma and that the item under consideration is from $W^*$, it follows that this utility matching will be achieved at a price of $2\epsilon$ or larger. At this point, the item is either already sold, or by dropping the price an $\epsilon$ more, it definitely gets sold, improving the revenue of seller $i$. This proves the claim.
\end{proof}

There is a close connection between the lexicographical preference ordering of sellers, and the order in which the sellers are approached by the buyer. We explore this in the next claim.
\begin{claim}
\label{cl:lexGreedy}
For any two sellers $i,j$ that are approached by the buyer, if $i < j <\ell$, or if $\ell < i < j$, then $i$ is approached before $j$. 
Additionally, every item sold by $i$ offers the same utility as every item sold by $j$.
\end{claim}
\begin{proof}
Consider sellers $i,j$ that were approached by the buyer, and say $i< j < \ell$. Suppose $j$ was approached before $i$. This is possible only if $j$ offered a strictly larger utility than $i$. In such a case, $j$ could reduce its utility to that offered by $i$. This makes $j$ get more revenue. This follows by noting that $j$ still manages to sell all the items he previously sold because the seller $\ell$ who offers an equal or lesser utility (because $\ell$ is approached later) is less preferred than $j$ yet is still selling at least one item. Similar argument holds for $i,j > \ell$ too except that to argue that $j$ manages to sell all items he previously sold, we must note that $i$ must have offered strictly larger utility to have been considered before $\ell$. Thus even if $j$ matches $i$'s utility, he still manages to be considered before $\ell$ because of higher utility offered to the buyer.

From the argument above, it follows that utility offered by $i$'s sold items is at least that offered by $j$'s sold items.  It is easy to see that seller $i$'s sold items will not offer a strictly higher utility than that of $j$ because even if $i$ matches $j$'s utility, $i$'s items will be considered before $j$'s.
\end{proof}

\begin{corollary}
\label{cl:equalUtilityGlobal}
All the sold items from sellers in the set $\{1,\dots,\ell-1\}$ offer the same utility, and the sold items from sellers in $\{\ell+1,\dots,s\}$ offer the same utility.
\end{corollary}
%\begin{proof}
%From Claim~\ref{cl:equalUtilityLocal} it follows that all sold items from the same seller offer the same utility. Consider two sellers $i$, $i+t$ such that $i < i+t <\ell$ that sell at least one item each. From Claim~\ref{cl:lexGreedy} it follows that utility offered by $i$'s sold items is at least that offered by $i+t$'s sold items.  It is easy to see that seller $i$'s sold items will not offer a strictly higher utility than that of $i+t$ because even if $i$ matches $i+t$'s utility, $i$'s items will be considered before $i+t$'s. Same argument holds for $\ell < i < i+t$.
%\end{proof}

\begin{claim}
\label{cl:uue}
Let $u$ be the utility offered by seller $\ell$. Then every sold item offers utility $u$ or $u+\epsilon$.
%Then the utilities offered by the sold items from sellers $1,\dots,\ell-1$ and that offered by sold items from sellers $\ell+1,\dots,s$ will be in $\{u,u+\epsilon\}$.
\end{claim}
\begin{proof}
% \moshecomment{Isn't there a simple proof just saying that if there are two items that are sold for which there is a gap of $2\epsilon$ in utility, then the cheap item can still be sold after an increase of $\epsilon$ in price, and all other items from the same seller will still be sold, increasing his utility?}
%\balucomment{The argument is essentially that. But some case analysis is needed because generally speaking, the utilities of the three sets $\{1,\dots,\ell-1\},\ \{\ell+1,\dots,s\},\ \{\ell\}$ could belong to $\{u+2\epsilon, u+\epsilon,u\}$. That is, to avoid losing due to tie breaking, it might be necessary to offer strictly more utility. To say that the $u+2\epsilon$ is not possible we should consider the two different orderings possible and say that in each of them the $u+2\epsilon$ is ruled out because of lexicographical preferences. The proof becomes long because of defining notations like $\ell^-$ and $\ell^+$, which are really proxies for $\ell-1$ and $s$: the latter cannot be used because they are perhaps not approached by the buyer.}
Suppose no seller was approached before $\ell$, then the claim is trivial. Let $\ell^-$ denote the largest numbered seller in $\{1,\dots,\ell-1\}$ who sells at least one item, and let $\ell^+$ be the largest numbered seller in $\{\ell+1,\dots,s\}$ who sells at least one item. By Claim~\ref{cl:lexGreedy}, the seller who was approached just before $\ell$ will be either $\ell^+$ or $\ell^-$. Suppose it was $\ell^+$, then seller $\ell^+$ should offer a strictly larger utility than $u$, but need not offer more than  $u+\epsilon$. Thus, by Corollary~\ref{cl:equalUtilityGlobal}, all sold items from sellers $\ell+1,\dots, s$ offer utility $u+\epsilon$. Given this, seller $\ell^-$ need not offer a utility higher than $u+\epsilon$, and he will not offer a smaller utility either, for otherwise he would have been considered after $\ell^+$. This completes the first case. Consider the other case where the seller approached just before $\ell$ was $\ell^-$. Then $\ell^-$ will offer a utility of exactly $u$, and by Corollary~\ref{cl:equalUtilityGlobal}, so will all the sold items from sellers $\{1,2,\dots,\ell-1\}$. The seller $\ell^+$ should offer utility $u+\epsilon$ to be considered before $\ell^-$, and by Corollary~\ref{cl:equalUtilityGlobal}, so will all the sold items from sellers $\ell+1,\dots,s$.
\end{proof}

Now consider seller $\ell$. He currently sells the most valuable set of $|S_{\ell}|$ items $I_\ell$ (if not he can raise his revenue by switching to sell more valuable items at an increased price while giving the buyer the same utility).
Suppose all the items from $W^*\cap I_{\ell}$ are sold, then the Theorem stands proved by Claim~\ref{cl:losingWelfareEll}. If $\ell$ had unsold items in $W^*\cap I_{\ell}$ of value at most $\underline{v}+\epsilon$, each such item causes at most an additive $\epsilon$ welfare loss, which our Theorem accommodates. Thus, assume that all unsold items in $W^*\cap I_{\ell}$ are of value at least $\underline{v}+2\epsilon$. Without loss of generality, $\ell$ can adjust the prices of all unsold items in $W^*\cap I_\ell$ so that they all offer the same utility $u$ offered by the sold items in $S_\ell$. This is because $\ell$ can drop the prices of unsold items in $W^*\cap I_\ell$, and along the lines of the argument in Claim~\ref{cl:losingWelfareEll}, at a price no smaller than $2\epsilon$ the utility of unsold items will match the utility $u$ offered by items in $S_\ell$. Now, by claim~\ref{cl:uue}, all sold items offer utility of $u$ or $u+\epsilon$. Seller $\ell$ has the option of dropping the prices of all items in $W^*\cap I_\ell$ by $2\epsilon$, to make their utility $u+2\epsilon$, there by selling all items in $W^*\cap I_\ell$. Seller $\ell$ chose not to do this implying that $\ell$ does not see any increase in revenue from doing this change. This means:

\begin{align}
\label{eqn:heteroEq}
\sum_{j \in W^* \cap (I_\ell\setminus S_\ell)} (v_j - u) - 2\epsilon \lvert W^*\cap I_\ell\rvert \leq 0
\end{align}
where $v_j$ is the buyer's value for item $j$.

Recall that $\underline{v}$ denote the value of the least valued item that gets sold in $ALG$. Since $u$ is the smallest utility that is offered in $ALG$, it follows that $u \leq \underline{v}$. This, combined with~\eqref{eqn:heteroEq} gives
\begin{align}
\label{eqn:heteroEq2}
\sum_{j \in W^* \cap (I_\ell\setminus S_\ell)} v_j \leq  \sum_{j \in W^* \cap (I_\ell\setminus S_\ell)} \underline{v} + 2\epsilon \lvert W^*\cap I_\ell\rvert
\end{align}
From Claim~\ref{cl:losingWelfareEll}, it follows that for all sellers other than $\ell$, we have
\begin{align}
\label{eqn:heteroEq3}
\sum_{i \neq \ell} \sum_{j \in W^* \cap (I_i\setminus S_i)} v_j \leq  \sum_{i\neq \ell} \sum_{j \in W^* \cap (I_i\setminus S_i)} (\underline{v} + \epsilon)
\end{align}

The optimal welfare can be written as
\begin{align*}
OPT &= \sum_i \sum_{j \in W^*} v_j\nonumber \\
&= \sum_i \left( \sum_{j\in W^* \cap S_i} v_j \right) + \sum_i \left(\sum_{j\in W^*\cap (I_i\setminus S_i)} v_j\right)\\
&\leq \sum_i \left( \sum_{j\in W^* \cap S_i} v_j \right) +\sum_{i} \sum_{j \in W^* \cap (I_i\setminus S_i)}\underline{v} + 2k\epsilon\\
&\leq ALG + 2k\epsilon
\label{eqn:heteroOPT}
\end{align*}
The first inequality in the chain above follows just by combining~\eqref{eqn:heteroEq2} and~\eqref{eqn:heteroEq3}, and the fact that $|W^*| \leq k$. The second follows from the definition of $\underline{v}$, namely, all sold items in $ALG$ are of value at least $\underline{v}$. \qed
\end{proof}

\end{document}